\documentclass{svmult}

\usepackage{amssymb}        
\usepackage{amsmath}        
\usepackage{epsfig}        
\usepackage{psfrag}

\begin{document}

\title*{On the radar method in general-relativistic spacetimes}
\author{Volker Perlick}
\institute{TU Berlin, Institute of Theoretical Physics \\
Sekr. PN 7-1, Hardenbergstrasse 36, 10623 Berlin, Germany \\  
\texttt{vper0433@itp.physik.tu-berlin.de}}

\maketitle

\begin{abstract}
If a clock, mathematically modeled by a parametrized timelike curve
in a general-relativistic spacetime, is given, the radar method 
assigns a time and a distance to every event which is sufficiently
close to the clock. Several geometric aspects of this method are
reviewed and their physical interpretation is discussed.
\end{abstract}

\section{Introduction}
\label{sec:intro}
When Einstein was asked about the meaning of time he used to say: ``Time
is the reading of a clock''. Taking this answer seriously, one is forced 
to accept that time is directly defined only at the position of a clock;
if one wants to assign a time to events at some distance from the clock
one needs an additional prescription. As such prescription, Einstein
suggested the radar method with light rays.

Although originally designed for special relativity, the radar method
works equally well in general relativity. What one needs is a clock
in an arbitrary general-relativistic spacetime. Here and in the following,
our terminology is as follows. A \emph{general-relativistic spacetime}
is a 4-dimensional manifold $M$ with a smooth metric tensor field $g$ of 
Lorentzian signature and a time orientation; the latter means that a 
globally consistent distinction between future and past has been made.
A \emph{clock} is a smooth embedding $\gamma : t \mapsto \gamma (t)$ 
from a real interval into $M$ such that the tangent vector ${\dot{\gamma}}(t)$ 
is everywhere timelike with respect to $g$ and future-pointing.
This terminology is justified because we can interpret the value of the 
parameter $t$ as the reading of a clock. Note that our definition of a 
clock does not demand that ``its ticking be uniform'' in any sense. 
Only smoothness and monotonicity is required.  

The radar method assigns a time and a distance to an event $q$ in the 
following way. One has to send 
a light ray from an event on the curve $\gamma$, say $\gamma (t_1)$, to 
$q$ and receive the reflected light ray at another event on $\gamma$, say
$\gamma (t_2)$, see Figure \ref{fig:radneigh}. The \emph{radar time} $T$ 
and the \emph{radar distance} $R$ of the event $q$ with respect to $\gamma$
are then defined by 
\begin{equation}\label{eq:defT}
  T= \frac{1}{2} \big( t_2+t_1 \big) \; ,
\end{equation} 
\begin{equation}\label{eq:defR}
  R= \frac{1}{2} \big( t_2-t_1 \big) \; .
\end{equation} 
Here and in the following, ``light ray'' tacitly means ``freely propagating
light ray'', i.e., it is understood that there is no optical medium and that
mirrors or other appliances that deviate a light ray are not used. 
Adopting the standard formalism of general relativity, ``light ray'' is then 
just another word for ``lightlike geodesic of the spacetime metric $g$''.

\begin{figure}
    \psfrag{p}{$p$} 
    \psfrag{q}{$q$} 
    \psfrag{U}{$U$} 
    \psfrag{V}{$V$} 
    \psfrag{x}{$\gamma (t_1)$} 
    \psfrag{y}{$\gamma (t_2)$} 
\centerline{\epsfig{figure=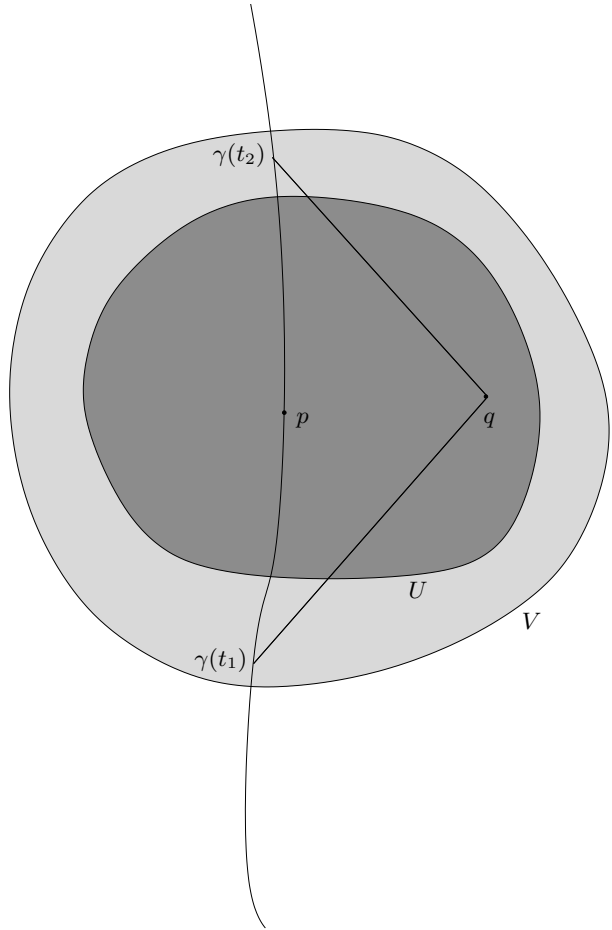, width=8cm}}
\caption{The radar method.}\label{fig:radneigh}
\end{figure}

In the following we discuss the radar method from a geometrical point of
view, reviewing some known results and formulating a few new ones. The radar
method has obvious relevance for the communication with satellites in the
solar system, because all such communication is made with the help of 
electromagnetic radiation that can be modeled, in almost all cases, in
terms of light rays. By sending a light ray to a satellite and receiving
the reflected signal the radar time $T$ and the radar distance $R$ of events
at the satellite are directly measurable quantities. Note that we do not need 
an experimentalist at the event $q$ where the light ray is reflected; a 
passive reflecting body, such as the LAGEOS satellites, would do. 

\section{Radar neighborhoods}
\label{sec:radneigh}
The radar time $T$ and the radar distance $R$ of an event $q$ with respect to 
a clock $\gamma$ are well-defined if there is precisely one future-pointing and 
precisely one past-pointing light ray from $q$ to $\gamma$. Neither existence 
nor uniqueness of such light rays is guaranteed.

It is possible that an event $q$ cannot be connected to $\gamma$ by any 
future-pointing (or any past-pointing) light ray. There are two physically
different situations in which this occurs: First, $q$ may be in a ``shadow'' 
of some obstacle that lies in the direction to $\gamma$; second, $q$ may be 
behind an ``event horizon'' of $\gamma$, see Figure \ref{fig:shadhor}.

\begin{figure}
    \psfrag{I}{I} 
    \psfrag{II}{II} 
    \psfrag{III}{III} 
    \psfrag{g}{$\gamma$} 
\centerline{\epsfig{figure=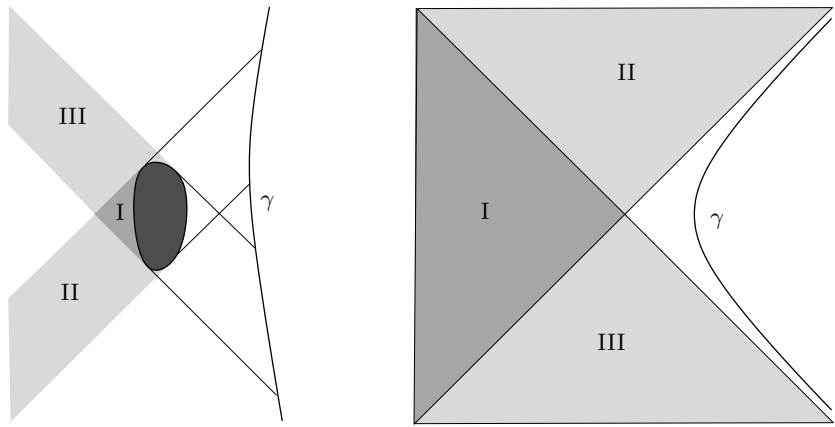, width=11cm}}
\caption{Shadows (left) and horizons (right) are obstacles for the radar method.
The example on the left shows a clock $\gamma$ in Minkowski spacetime with a subset
removed. The example on the right shows a clock $\gamma$ with uniform acceleration in
Minkowski spacetime. In both cases, events in the region II cannot be connected
to $\gamma$ by a future-pointing light ray, events in the region III cannot be
connected to $\gamma$ by a past-pointing light ray, and events in the region I
cannot be connected to $\gamma$ by any light ray.}\label{fig:shadhor}
\end{figure}

It is also possible that an event $q$ can be connected to $\gamma$ by two or
more future-pointing (or past-pointing) light rays. Whenever
the future light cone (or the past light cone) of $q$ has a caustic or a 
transverse self-intersection, it meets some timelike curves at least 
twice, see \cite{Perlick1996} or \cite{Perlick2004} for a detailed 
discussion. If the past light cone of $q$ intersects $\gamma$ at least 
twice, an observer at $q$ sees two or more images of $\gamma$, 
i.e., we are in a gravitational lensing 
situation. Figure \ref{fig:cone} shows an example of a past light cone
that has two intersections with appropriately chosen timelike curves,
as is geometrically evident from the picture. 

These observations clearly show that, in an arbitrary general-relativistic 
spacetime, the radar method does not work globally. However, it always
works locally. This is demonstrated by the following simple proposition.

\begin{proposition}\label{prop:radneigh}
Let $\gamma$ be a clock in an arbitrary general-relativistic 
spacetime and $p= \gamma (t_0)$ some point on $\gamma$. Then there are open
subsets $U$ and $V$ of the spacetime with $p \in U \subset V$ such that every
point $q$ in $U \setminus \mathrm{image}(\gamma)$ can be connected to the 
worldline of $\gamma$ by
precisely one future-pointing and precisely one past-pointing light ray 
that stays within $V$, see Figure $\ref{fig:radneigh}$. In this case, $U$ is called
a \emph{radar neighborhood} of $p$ with respect to $\gamma$.
\end{proposition}
To prove this, we just have to recall that every point in a general-relativistic
spacetime admits a convex normal neighborhood, i.e., a neighborhood $V$  
such that any two points in $V$ can be connected by precisely one geodesic 
that stays within $V$. Having chosen such a $V$, it is easy to verify that every 
sufficiently small neighborhood $U$ of $p$ satisfies the desired property.     

\begin{figure}
    \psfrag{g}{$\gamma$} 
    \psfrag{R}{$R= \mathrm{constant}$} 
    \psfrag{T}{$T= \mathrm{constant}$} 
\centerline{\epsfig{figure=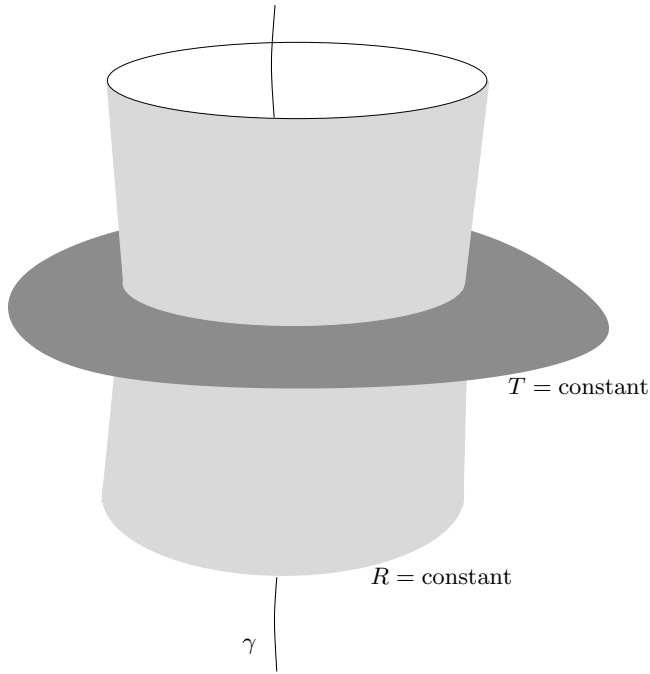, width=8cm}}
\caption{Hypersurfaces $T=\mathrm{constant}$ and hypercylinders $R= \mathrm{constant}$
defined by the radar method.}
\label{fig:radco}
\end{figure}

As an aside, we mention that the existence of radar neighborhoods, in the sense
of Proposition \ref{prop:radneigh}, was chosen as one of the
axioms in the axiomatic approach to spacetime theory by Ehlers, Pirani and 
Schild \cite{EhlersPiraniSchild1972}.

If $U$ is a radar neighborhood, the radar time $T$ and the radar distance $R$
are well-defined functions on $U \setminus \mathrm{image}(\gamma)$. By continuous extension
onto the image of $\gamma$ one gets smooth hypersurfaces $T= \mathrm{constant}$ that 
intersect $\gamma$ orthogonally; hence, they are spacelike near $\gamma$. Note, 
however, that they need not be spacelike on the whole radar neighborhood. The 
hypersurfaces $R= \mathrm{constant}$ have a cylindrical topology, see 
Figure \ref{fig:radco}. Incidentally, if one replaces (\ref{eq:defT}) by 
$T = p \, t_1 + (1-p) \, t_2$ with any number $p$ between 0 and 1, each hypersurface
$T= \mathrm{constant}$ gets a conic singularity at the intersection point with 
$\gamma$. This clearly shows that the choice of the factor 1/2 is the most natural
and the most convenient one. (If one allows for \emph{direction-dependent} factors,
one can get smooth hypersurfaces with factors other than 1/2. This idea, which 
however seems a little bit contrived, was worked out by Havas \cite{Havas1987} 
where the reader can find more on the ``conventionalism debate'' around the factor 
1/2.)

By covering $\gamma$ with radar neighborhoods $U$ (and the pertaining convex normal
neighborhoods $V$), it is easy to verify that $T$ and $R$ coincide on the intersection
of any two radar neighborhoods. Hence, $T$ and $R$ are well-defined on some tubular 
neighborhood of $\gamma$. We will now investigate how large this neighborhood can
be for the case of a clock moving in the Solar System, the latter being modeled
by the Schwarzschild spacetime around the Sun. 

To that end we consider the Schwarzschild spacetime around a non-transparent spherical
body of radius $r_*$ and mass $m$. (The radius is measured in terms of the radial
Schwarzschild coordinate and for the mass we use geometrical units, i.e., the
Schwarzschild radius is $2m$.) Using the standard deflection formula for light rays
in the Schwarzschild spacetime, the following result can be easily verified. If a bundle of 
light rays comes in initially parallel from infinity, the rays that graze the surface
of the central body will meet the axis of symmetry of the bundle at radius 
\begin{equation}\label{eq:defrf}
  r_f= \frac{r_*}{4 \frac{m}{r_*} + O\big( (\frac{m}{r_*})^2 \big)}
  \approx \frac{r_*^2}{4 m}
\end{equation}
see Figure \ref{fig:rf}. This radius $r_f$ is sometimes called the \emph{focal length}
of a non-transparent body of radius $r_*$ and mass $m$. If we insert the values of
our Sun we find
\begin{equation}\label{eq:rfsun}
r_f \approx 550 \; \mathrm{a.u.} 
\end{equation}
where 1 a.u. = 1 astronomical unit is the average distance from the Earth to the Sun. 
From any event at $r<r_f$, the future-pointing and past-pointing light rays spread 
out without intersecting each other. They cover the whole space $r> r_*$ with the 
exception of those points that lie in the ``shadow'' cast by the central 
body, see Figure \ref{fig:rf}. By contrast, light rays from an event at $r>r_f$ do 
intersect; the past light cone of such an event is shown in Figure \ref{fig:cone}.

\begin{figure}
    \psfrag{f}{$r_f$} 
\centerline{\epsfig{figure=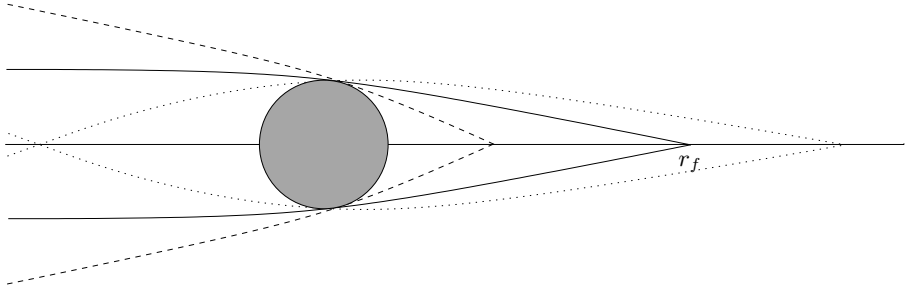, width=12cm}}
\caption{The focal length $r_f$ of a non-transparent spherical body.}
\label{fig:rf}
\end{figure}

As a consequence, for a clock $\gamma$ moving arbitrarily in the region $r> r_*$, an 
event $q$ at a radius $r$ with $r_*<r<r_f$ can be connected to the worldline of 
$\gamma$ by at most one future-pointing and at most one past-pointing light ray.  
We shall make the additional assumption that $\gamma$ is inextendible and 
approaches neither the surface of the central body nor infinity in the future or
in the past. This assures that there are no event horizons for $\gamma$. As a consequence,
any event $q$ at radius $r$ with $r_* < r < r_f$ can be connected to $\gamma$ by 
precisely one future-pointing and precisely one past-pointing light ray unless $\gamma$ 
moves through the shadow cast by the central body for light rays issuing from $q$.

An event $q$ at radius $r>r_f$, on the other hand, may be connected to the
worldline of a clock by several future-pointing (or past-pointing) light
rays. This is geometrically evident from Figure \ref{fig:cone}.

\begin{figure}
    \psfrag{f}{$r_f$} 
\centerline{\epsfig{figure=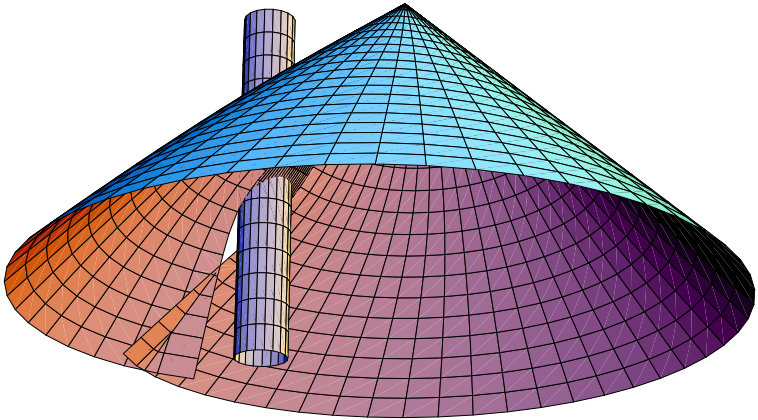, width=10cm}}
\caption{Past light cone of an event at radius $r > r_f$ in the field of a non-transparent
gravitating body. The ``chimney'' is the world tube of the gravitating body. The ``shadow''
is clearly seen as a gap in the light cone. In this $2+1$-dimensional picture the light cone
forms a transverse self-intersection. If the gravitating body is spherical symmetric, taking
the missing spatial dimension into account shows that actually a sphere's worth of light
rays is focused at each point of the intersection set.}
\label{fig:cone}
\end{figure}

So for any clock in the Solar System the radar method assigns a unique time $T$ and 
a unique distance $R$ to any event at radius $r<r_f$, with the exception of those events
for which the clock lies in the shadow of the central body. Note that for all existing 
spacecraft the distance from the Sun is considerably smaller than $r_f=550$ a.u. 
(In October 2005, the spacecraft farthest away from the Sun was Pioneer 10 with a 
distance of 89 a.u.)

The idea of sending a spacecraft to $r> 550$ a.u. was brought forward by Eshleman
\cite{Eshleman1979} in 1979. What makes this idea attractive is the possibility of
observing distant light sources strongly magnified by the focusing effect of the 
gravitational field of the Sun, see again Figure \ref{fig:rf}. For a detailed 
discussion of the perspectives of such a mission see 
Turyshev and Andersson \cite{TuryshevAndersson2003}. 

It should be emphasized that our consideration applies only to a non-transparent body.
If the central body is transparent, light rays passing through the central region of the 
body are focussed at a radius that is much smaller than the $r_f$ given above. If the 
interior is modeled by a perfect fluid with constant density, one finds for the Sun a 
focal length of 30 a.u., in comparison to the 550 a.u. for the non-tranparent case, see
Nemiroff and Ftaclas \cite{NemiroffFtaclas1997}. A transparent Sun is a reasonable
model for neutrino radiation (which travels approximately, though not precisely, on lightlike
geodesics) and for gravitational radiation (which travels along lightlike geodesics if 
modeled as a linear perturbation of the Schwarzschild background). So, the focusing 
at 30 a.u. might have some futuristic perspective in view of neutrino astronomy and 
gravitational wave astronomy.  
 
\section{Characterization of standard clocks with the radar method}
\label{sec:standard}
If we reparametrize the curve $\gamma$, the hypersurfaces $T=\mathrm{constant}$ and $R=\mathrm{constant}$
change. Therefore, the radar method can be used to characterize distinguished parametrizations
of worldlines, i.e., distinguished clocks. In a general-relativistic spacetime, the standard clock
parametrization is defined by the condition
\begin{equation}\label{eq:standard}
\frac{d}{dt}  g \big( \dot{\gamma} (t) , \dot{\gamma} (t) \big) =0
\end{equation}
where $g$ is the spacetime metric. This defines a parametrization along any timelike curve that 
is unique up to affine transformations, $t \mapsto at+b$ with real constants $a$ and $b$. As we 
restrict to future-pointing parametrizations, $a$ must be positive. Then the choice of $a$ 
determines the unit and the choice of $b$ determines the zero on the dial. By choosing $a$
appropriately, we can fix the unit of a standard clock such that $g \big( \dot{\gamma} (t) , 
\dot{\gamma} (t) \big) = -1$. Then the parameter of the clock is called \emph{proper time}.
Note that under an
affine reparametrization $t \mapsto at + b$ the radar time and the radar distance transform
according to $T \mapsto aT+b$ and $R \mapsto aR$, i.e., the hypersurfaces $T=\mathrm{constant}$ 
and $R=\mathrm{constant}$ are relabeled but remain unchanged.  

With the help of the radar method one can formulate an operational prescription that allows to
test whether a clock is a standard clock. This prescription is now briefly reviewed, for details and
proofs see \cite{Perlick1987}. Here we assume that the test is made in a general-relativistic
spacetime; in \cite{Perlick1987} the more general case of a Weylian spacetime is considered.

To test whether a clock $\gamma$ behaves like a standard clock in a particular
event $\gamma (t_0)$, we emit at this event two freely falling particles in spatially
opposite directions. These two freely falling particles are mathematically modeled by timelike 
geodesics $\mu$ and $\overline{\mu}$, and the condition that they are emitted in spatially 
opposite directions means that the future-oriented tangent vector to $\gamma$ is a convex 
linear combination of the future-oriented tangent vectors to $\mu$ and $\overline{\mu}$.
If we restrict to a radar neighborhood of $\gamma (t_0)$, the radar method assigns a 
time $T$ and a distance $R$ to each event on $\mu$, and a time
$\overline{T}$ and a distance $\overline{R}$ to each event on $\overline{\mu}$, see 
Figure \ref{fig:standard}. These quantities can be actually measured provided that the two 
freely falling particles are reflecting objects. From these measured quantities we can 
calculate the differential quotients $dR/dT$ and $d^2R/dT^2$ along $\mu$ and the differential
quotients $d\overline{R}/d\overline{T}$ and $d^2\overline{R}/d\overline{T}{}^2$ along $\overline{\mu}$,
i.e., the \emph{radar velocity} and the \emph{radar acceleration} of the two freely falling
particles. 
It is shown in \cite{Perlick1987} that the standard clock condition (\ref{eq:standard}) holds 
at $t = t_0$ (which corresponds to $T= \overline{T}=t_0$) if and only if
\begin{equation}\label{eq:test}
  \frac{\frac{d^2R}{dT^2}}{1- \big( \frac{dR}{dT} \big) ^2} \; \Big| _{T=t_0} \: = \: - \:   
  \frac{\frac{d^2\overline{R}}{d\overline{T}{}^2}}{1- 
  \big( \frac{d\overline{R}}{d\overline{T}} \big) ^2} \; \Big| _{\overline{T}=t_0} \: .
\end{equation}
This prescription can be used, in particular, to directly test whether atomic clocks
are standard clocks. All experiments so far are in agreement with this hypothesis, but
a direct test has not been made.

\begin{figure}
    \psfrag{z}{$\gamma (t_0)$} 
    \psfrag{m}{$\mu$} 
    \psfrag{n}{$\overline{\mu}$} 
    \psfrag{M}{$(R,T)$} 
    \psfrag{N}{$(\overline{R},\overline{T})$} 
\centerline{\epsfig{figure=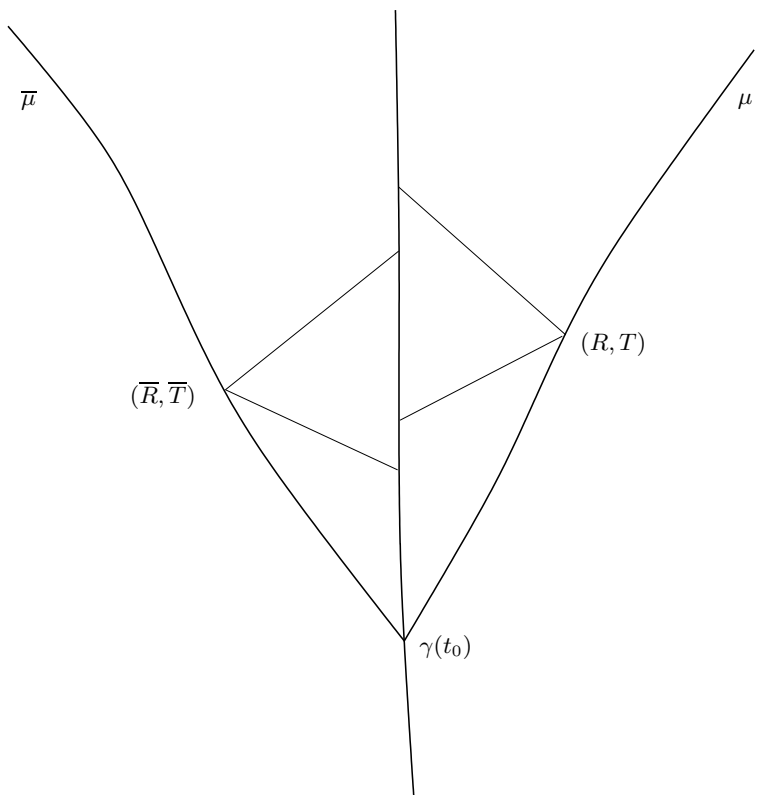, width=10cm}}
\caption{Testing a clock as a standard clock with the radar method.}
\label{fig:standard}
\end{figure}

There are alternative characterizations of standard clocks by Marzke and Wheeler
\cite{MarzkeWheeler1964} and Kundt and Hoffman \cite{KundtHoffman1962} which also
work with light rays and freely falling particles. The advantages of the method reviewed
here in comparison to these two older methods are outlined in \cite{Perlick1994}.

\section{Radar coordinates, optical coordinates, and Fermi coordinates}
\label{sec:Fermi}

Given any clock $\gamma$ in any general-relativistic spacetime, the radar method assigns, 
as outlined above, to each event $q$ in some tubular neighborhood of $\gamma$ a radar time 
$T$ and a radar distance $R$. In order to get a coordinate system (\emph{radar coordinates}) 
on this tubular neighborhood, we may add two angular coordinates $\vartheta$ and $\varphi$ 
in the following way. Choose at each point $\gamma (t)$ an orthonormal tetrad 
$(E_0(t),E_1(t),E_2(t),E_3(t))$, smoothly dependent on $t$, such that $E_0(t)$ is 
future-pointing and tangent to $\gamma$. To each event $q$ consider the past-oriented 
light ray, in the notation of Figure \ref{fig:radneigh}, from $\gamma (t_2)$ to $q$. 
The tangent vector to this light ray at $\gamma (t_2)$ must be proportional to a vector 
of the form $-E_0(t_2)+ \mathrm{cos} \, \varphi \, \mathrm{sin} \, \vartheta \, E_1 (t_2) 
+ \mathrm{sin} \, \varphi \, \mathrm{sin} \, \vartheta \, E_2 (t_2)+ 
\mathrm{cos} \, \vartheta \, E_3 (t_2)$ which defines $\vartheta$ and $\varphi$. 
Thus, $\vartheta$ and $\varphi$ indicate at which point on the sky of $\gamma$ 
the event $q$ is seen. Just as with ordinary spherical coordinates, there are coordinate
singularities at $R= 0$ and at $\mathrm{sin}\,\vartheta =0$, and $\varphi$ has to
be identified with $\varphi + 2 \pi$. Apart from these obvious pathologies, 
the radar coordinates $(T,R,\vartheta,\varphi)$ form a well-defined coordinate
system on some tubular neighborhood of $\gamma$. There are two possibilities of 
modifying the radar coordinates without changing the information contained in them.  
First, one may replace $T$ and $R$ by $t_1$ and $t_2$, according to (\ref{eq:defT}) 
and (\ref{eq:defR}), and use the modified radar coordinates $(t_1,t_2,\vartheta,\varphi)$. 
Second, one may switch to Cartesian-like coordinates $(T,x,y,z)$ by introducing 
$x= R \, \mathrm{cos} \, \varphi \, \mathrm{sin} \, \vartheta$, 
$y = R \, \mathrm{sin} \, \varphi \, \mathrm{sin} \, \vartheta$,
$z = R \, \mathrm{cos} \, \vartheta$ in order to remove the coordinate
singularities at $R=0$ and $\mathrm{sin} \, \vartheta = 0$. Radar coordinates 
have been used as a tool, e.g., in the axiomatic approach to spacetime theory
of Schr{\"o}ter and Schelb \cite{Schroeter1988, SchroeterSchelb1992, Schelb1992}.

We will now compare radar coordinates with two other kinds of coordinate systems that 
can be introduced near the worldline of any clock $\gamma \,$: ``optical coordinates'' and 
``Fermi coordinates''. We will see that there are some similarities but also major 
differences between these three types of coordinate systems. For an alternative 
discussion of optical coordinates and Fermi coordinates see Synge \cite{Synge1960}.

\emph{Optical coordinates} were introduced by Temple \cite{Temple1938}. The alternative
name \emph{observational coordinates} is also common, see Ellis et al.
\cite{Ellis1980, EllisNelMaartensStoegerWhitman1985}. They assign to the 
event $q$ the four-tuple $(t_2,s,\vartheta,\varphi)$, 
where $t_2$, $\vartheta$ and $\varphi$ have the same meaning as above and $s$ is the 
``affine length'' (or ``projected length'') along the past-oriented light ray from 
$\gamma (t_2)$ to $q$. Using the exponential map $\mathrm{exp}$ determined by the 
spacetime metric, $s$ can be defined by the equation 
\begin{equation}\label{eq:expo}
q \: = \: \mathrm{exp}_{\gamma (t_2)} \Big( s \big( -E_0(t_2)+
\mathrm{cos} \, \varphi \, \mathrm{sin} \, \vartheta \, E_1 (t_2) + 
\mathrm{sin} \, \varphi \, \mathrm{sin} \, \vartheta \, E_2 (t_2)+ 
\mathrm{cos} \, \vartheta \, E_3 (t_2)\big)\Big) \: . 
\end{equation}
Just as radar coordinates, optical coordinates are well-defined, apart from the obvious 
coordinate singularities at $s=0$ and $\mathrm{sin}\,\vartheta =0$ on some tubular 
neighborhood of $\gamma$. The boundary of
this neighborhood is reached when the past light cone of an event on $\gamma$ develops a 
caustic or a transverse self-intersection. (Beyond such points, the optical coordinates are
multi-valued. This does not mean that they are useless there; however, they do not define
a coordinate system in the usual sense.) As radar coordinates require a similar condition not 
only on past light cones but also on future light cones, the domain of radar coordinates 
is always contained in the domain of optical coordinates. Also, there is an important 
advantage of optical coordinates in view of calculations: Optical coordinates only require 
to calculate the past-pointing lightlike geodesics issuing from points on $\gamma$; radar 
coordinates require to calculate past-pointing and future-pointing lightlike geodesics 
from points on $\gamma$, and to determine their intersections. Nonetheless, for applications 
in the solar system radar coordinates are advantageous because they have an operational 
meaning. In principle, optical coordinates also have an operational meaning: $(t_2, \vartheta ,
\varphi )$ are the same as in radar coordinates, and for the affine (or projected)
length $s$ a prescription of measurement was worked out by Ruse \cite{Ruse1933} after 
this length measure had been introduced mathematically by Kermack, McCrea and Whittaker
\cite{KermackMcCreaWhittacker1932}. However, this prescription requires the distribution
of assistants with rigid rods along each light ray issuing from $\gamma$ into the past which 
is, of course, totally unrealistic in an astronomical situation. In this sense, optical
coordinates have an operational meaning only in principle but not in practice, whereas
radar coordinates have an operational meaning both in principle and in practice, at least
in the Solar System. In cosmology, however, this is no longer true. Then the radar 
coordinates, just as the optical coordinates, have an operational meaning only in 
principle but not in practice: Sending a light ray to a distant galaxy and waiting for
the reflected ray is a ridiculous idea. As a matter of fact, optical coordinates are
much more useful in cosmology than radar coordinates. Although $s$ is not directly 
measurable, it is related in some classes of spacetimes to other distance measures, 
such as the redshift or the angular diameter distance, which can be used to replace 
$s$. For applications of optical coordinates in cosmology see 
\cite{EllisNelMaartensStoegerWhitman1985}. As the simplest example, one may consider
optical coordinates and radar coordinates in Robertson-Walker spacetimes, cf. Jennison 
and McVittie \cite{JennisonMcVittie1975} and Fletcher \cite{Fletcher1994}.

We now turn to \emph{Fermi coordinates} which were introduced by Enrico Fermi \cite{Fermi1922}. 
Let us recall how they are defined. As above, we have to choose along $\gamma$ an orthonormal 
tetrad $(E_0(t),E_1(t),E_2(t),E_3(t))$ with $E_0$ tangent to $\gamma$. Following Fermi, we 
require that the covariant derivative of each spatial axis $E_{\mu}$ ($\mu=1,2,3$) is 
parallel to the tangent of $\gamma$. This \emph{Fermi transport} law can be operationally 
realized by means of gyroscope axes \cite{MisnerThorneWheeler1973} or Synge's \emph{bouncing 
photon} method \cite{Synge1960,Pirani1965}. (Actually, the construction below can be carried 
through equally well if the spatial axes are not Fermi parallel. What is needed is only smooth 
dependence on the foot-point, just as with radar coordinates and optical coordinates.)
Then every event $q$ in a sufficiently small tubular neighborhood of $\gamma$ can 
be written in the form 
\begin{equation}\label{eq:expf}
q \: = \: \mathrm{exp}_{\gamma (\tau )} \Big( \, \rho \, \big(
\mathrm{cos} \, \phi \, \mathrm{sin} \, \theta \, E_1 (\tau) + 
\mathrm{sin} \, \phi \, \mathrm{sin} \, \theta \, E_2 (\tau)+ 
\mathrm{cos} \, \theta \, E_3 (t_2)\big)\Big) \: . 
\end{equation}
The Fermi coordinates of the point
$q$ are the four numbers $(\tau , \rho , \theta , \phi )$. Thus, each surface
$\tau = \mathrm{constant}$ is generated by the geodesics issuing orthogonally from the
point $\gamma (\tau )$. The distance $\rho$ is defined analogously to the affine length
in the optical coordinates, but now along spacelike rather than lightlike geodesics. 
Also, the angular coordinates $\theta$ and $\phi$ are analogous to the angular coordinates
$\vartheta$ and $\varphi$ in the radar and optical coordinates, but now they indicate 
the direction of a spacelike vector, rather than the direction of the spatial part 
of a lightlike vector. Just as the other two 
coordinate systems, Fermi coordinates are well-defined only on some tubular neighborhood
of $\gamma$. There are two reasons that limit this neighborhood. First, a hypersurface 
$\tau = \mathrm{constant}$ might develop caustics or self-intersections. Second, two
hypersurfaces $\tau = \mathrm{constant}$ might intersect. In contrast to radar coordinates,
Fermi coordinates are insensitive to reparametrizations of $\gamma$ (apart from the fact
that the surfaces $\tau = \mathrm{constant}$ are relabeled). The difficulty involved in
their calculation is the same as for optical coordinates which is considerably less
than for radar coordinates, as already mentioned above. The essential drawback of Fermi
coordinates is in the fact that they have absolutely no operational meaning: None of the
four coordinates $\tau$, $\rho$, $\theta$ and $\phi$ can be measured because there is
no prescription for physically realizing a spacelike geodesic orthogonal to a worldline.

In spite of this fact, Fermi coordinates have found many applications because sometimes
physically relevant effects can be conveniently calculated in terms of Fermi coordinates.
For a plea in favor of Fermi coordinates, in comparison to radar coordinates, see
Bini, Lusanna and Mashhoon \cite{BiniLusannaMashhoon2005}. In Minkowski spacetime, e.g.,
it is fairly difficult to calculate the radar time hypersurfaces $T =  \mathrm{constant}$
for an accelerating clock. By contrast, the Fermi time hypersurfaces $\tau = \mathrm{constant}$
are just the hyperplanes perpendicular to the worldline which are quite easy to determine.
(Of course, for an accelerating clock these hyperplanes necessarily intersect, so they
cannot form a smooth foliation on all of Minkowski spacetime.) It is an interesting
question to ask for which clocks the radar time hypersurfaces $T = \mathrm{constant}$
coincide with the Fermi time hypersurfaces $\tau = \mathrm{constant}$. For standard
clocks (recall Section \ref{sec:standard}) in Minkowski spacetime, Dombrowski, Kuhlmann
and Proff \cite{DombrowskiKuhlmannProff1984} have found the following answer.

\begin{proposition}\label{prop:coin}
Let $\gamma$ be a standard clock in Minkowski 
spacetime. Then the following two statements are equivalent.
\begin{itemize}
\item[\emph{(a)}] The radar time hypersurfaces $T = \mathrm{constant}$ are hyperplanes, i.e., 
they coincide with the Fermi time hypersurfaces $\tau = \mathrm{constant}$.
\item[\emph{(b)}] The $4$-acceleration of $\gamma$ is constant $($i.e., 
a Fermi-transported vector along $\gamma)$.
\end{itemize}
\end{proposition}

A worldline with constant 4-acceleration in Minkowski spacetime is either a straight line
(``inertial observer'', for which the 4-acceleration is zero) or a hyperbola (``Rindler observer'',
for which the 4-acceleration is a non-zero Fermi-transported vector, see Figure \ref{fig:shadhor}). 
It is easy to check that, indeed, in both cases the radar time hypersurfaces with respect 
to proper time parametrization are hyperplanes. The non-trivial 
statement of Proposition \ref{prop:coin} is in the fact that these are the \emph{only} cases 
for which the radar time hypersurfaces are hyperplanes.

We end this section with a remark on the fact that the term ``radar coordinates'' has been
used in the literature also in another way. Instead of supplementing
the radar time $T$ and the radar distance $R$  with
two angular coordinates, one could choose a second clock $\tilde{\gamma}$ which defines a 
radar time $\tilde{T}$ and a radar distance $\tilde{R}$. If the two clocks are sufficiently
close, $(T,R,\tilde{T},\tilde{R})$ can be used as coordinates on some open subset which
is not a tubular neighborhood of either clock.  Of course, one can replace $(T,R)$ by 
$(t_1,t_2)$ according to (\ref{eq:defT}) and (\ref{eq:defR}), and analogously $(\tilde{T},
\tilde{R})$ by $(\tilde{t}{}_1,\tilde{t}{}_2)$. In the coordinates
$(t_1,t_2,\tilde{t}{}_1,\tilde{t}{}_2)$, which are used e.g. by Ehlers, Pirani and Schild
\cite{EhlersPiraniSchild1972}, the coordinate hypersurfaces are light cones. Thus, the
construction makes use of the fact that four light cones generically intersect in a 
point. (Two light cones generically intersect in a 2-dimensional manifold, where 
``generically'' means that we have to exclude points where one of the light cones 
fails to be a submanifold and points where the two light cones are tangent. Similarly, 
three light cones generically intersect in a 1-dimensional manifold.) In this sense the
radar coordinates of Ehlers, Pirani and Schild are similar to the GPS type coordinates
of Blagojevi{\'c}, Garecki, Hehl and Obukhov \cite{BlagojevicGareckiHehlObukhov2002}.
The only difference is that the latter characterize each point as intersection
of four future light cones that issue from four given worldlines (``GPS satellites''),
whereas the former characterize each point as intersection of two future and two past
light cones that issue from two given worldlines.

\section{Synchronization of clocks}
\label{sec:synchronization}

Let $\gamma$ be a clock in an arbitrary general-relativistic spacetime, and consider  a
second clock $\tilde{\gamma}$. If $\tilde{\gamma}$ is sufficiently close to $\gamma$, the 
radar method, carried through with respect to the clock $\gamma$, assigns a unique time 
$T(\tilde{t})$ and a unique distance $R(\tilde{t})$ to each event $\tilde{\gamma} (\tilde{t})$. 
We say that $\tilde{\gamma}$ is \emph{synchroneous} to $\gamma$ if $T(\tilde{t})=\tilde{t}$ for
all $\tilde{t}$ in the considered time interval. (Instead of synchroneous one may say 
\emph{Einstein synchroneous} or \emph{radar synchroneous} to be more specific.) Clearly,
for every worldline sufficiently close to $\gamma$ there is a unique parametrization that
is synchroneous to $\gamma$. Selecting this particular parametrization is called 
\emph{synchronization} with $\gamma$. Note that the relation of being synchroneous 
is not symmetric: $\tilde{\gamma}$ may be synchroneous to $\gamma$ without $\gamma$ being 
synchroneous to $\tilde{\gamma}$. As an example, we may choose two affinely parametrized
straight timelike lines $\gamma$ and $\tilde{\gamma}$ in Minkowski spacetime that are not
parallel. If we arrange the parameters such that $\tilde{\gamma}$ is synchroneous to
$\gamma$, the converse is not true. Also, the relation of being synchroneous 
is not transitive: If $\tilde {\gamma}$ is synchroneous to $\gamma$ and $\hat{\gamma}$ is 
synchroneous to $\tilde{\gamma}$, it is not guaranteed that $\hat{\gamma}$ is synchroneous 
to $\gamma$. This non-transitivity is best illustrated with the \emph{Sagnac effect}: 
Consider a family of clocks along the rim of a rotating circular platform in Minkowski 
spacetime. Starting with any one of these clocks, synchronize each clock  
with its neighbor on the right. Then there is a deficit time interval after 
completing the full circle.

The following proposition characterizes the special situation that two clocks are mutually
synchroneous. 

\begin{proposition}\label{prop:symmetry}
Let $\gamma : \mathbb{R} \longrightarrow M$ and $\tilde{\gamma} : \mathbb{R} \longrightarrow M$ 
be two clocks, in an arbitrary spacetime, for which
the parameter extends from $- \infty$ to $\infty$. Assume that the worldlines of the two clocks
have no intersection but are sufficiently close to each other such that the radar method can be 
carried through in both directions. If $\tilde{\gamma}$ is synchroneous to $\gamma$
and $\gamma$ is synchroneous to $\tilde{\gamma}$, then the radar distance $R$ of $\tilde{\gamma}$ 
with respect $\gamma$ is a constant $R_0$, and the radar distance $\tilde{R}$ of $\gamma$
with respect to $\tilde{\gamma}$ is the same constant $R_0$.
\end{proposition}
\begin{proof}
The radar method carried through with $\gamma$ assigns to each event $\tilde{\gamma} ( \tilde{t} )$
a time $T(\tilde{t})$ and a distance $R(\tilde{t})$. Analogously, the radar method carried
through with respect to $\tilde{\gamma}$ assigns to each event $\gamma (t)$ a time $\tilde{T}
(t)$ and a distance $\tilde{R} (t)$. This implies the following identities, see Figure
\ref{fig:symmetry}. 
\begin{equation}\label{eq:synch}
\begin{split}
	t = T \big( \tilde{T}(t)- \tilde{R} (t) \big) +
      R \big( \tilde{T}(t)- \tilde{R} (t) \big) \; ,
\\
	\tilde{t} = \tilde{T} \big( T(\tilde{t})- R (\tilde{t}) \big) +
      \tilde{R} \big( T(\tilde{t})- R (\tilde{t}) \big) \; .
\end{split}
\end{equation}
If the clocks are mutually synchroneous, $T$ and $\tilde{T}$ are the identity maps, so 
(\ref{eq:synch}) simplifies to
\begin{equation}\label{eq:symmetry}
\begin{split}
	\tilde{R} (t) =
      R \big( t- \tilde{R} (t) \big) \; ,
\\
	R (\tilde{t}) =
      \tilde{R} \big( \tilde{t}- R (\tilde{t}) \big) \; .
\end{split}
\end{equation}
These equations hold for all $t$ and all $\tilde{t}$ in $\mathbb{R}$. By considering the
special case $t = \tilde{t} -R( \tilde{t})$ we find
\begin{equation}\label{eq:period}
	R(\tilde{t}) = R \big( \tilde{t}-2 R(\tilde{t}) \big) 
\end{equation} 
for all $\tilde{t}$ in $\mathbb{R}$. To ease notation, we drop the tilde in the following.
By induction, (\ref{eq:period}) yields
\begin{equation}\label{eq:induct}
	R(t) = R \big( t-2 n R(t) \big) \qquad \text{for all} \: n \in \mathbb{N} \, .
\end{equation} 
It is now our goal to prove that (\ref{eq:induct}) implies that $R$ is a constant. By
contradiction, assume there is a point where $R$ has negative derivative, $R'( t_*)<0$.
Then we must have
\begin{equation}\label{eq:limit}
	\frac{t_* \, - \, (t_*+ \varepsilon ) \, + \, 2 \, R(t_* + \varepsilon )}{
	2 \big( R(t_*)-R(t_*+ \varepsilon ) \big)} \, \longrightarrow \infty \qquad
	\text{for} \: \varepsilon \to +0 \; ,
\end{equation} 
because, by our assumption that the worldlines of the two clocks do not intersect,
$R(t_*)>0$.  Thus, there is an infinite sequence $t_n$ that converges towards
$t_*$ from above, such that
\begin{equation}\label{eq:sequence}
	n \: = \:
	\frac{t_* \, - \, t_n \, + \, 2 \, R(t_n )}{
	2 \big( R(t_*)-R(t_n ) \big)}  \qquad
	\text{for all sufficiently large} \: \; n \in \mathbb{N} \; .
\end{equation} 
As (\ref{eq:sequence}) can be rewritten as
\begin{equation}\label{eq:rational}
	t_n \, - \, 2 \, (n+1) \, R(t_n )  \; = \; 
	t_* \, - \, 2 \, n \, R(t_* ) \; , 
\end{equation} 
our earlier result (\ref{eq:induct}) yields $R(t_n)=R(t_*)$ for all members $t_n$
of our sequence, which obviously contradicts the assumption $R'(t_*)<0$. We have
thus proven that $R'(t) \ge 0$ for all $t$. But then we must have $R'(t)=0$ for all
$t$, because, again by (\ref{eq:induct}), to every $t$ there is a smaller parameter
value at which the function $R$ takes the same value. Hence, $R$ must be a constant,
$R(t)=R_0$ for all $t$. It is obvious from (\ref{eq:symmetry}) that then $\tilde{R}$
must take the same constant value. 
\end{proof}

\begin{figure}
    \psfrag{Q}{$T(\tilde{t})+R(\tilde{t})$} 
    \psfrag{P}{$\tilde{t}$} 
    \psfrag{S}{$T(\tilde{t})-R(\tilde{t})$} 
    \psfrag{E}{$\tilde{T}(t)+\tilde{R}(t)$} 
    \psfrag{F}{$t$} 
    \psfrag{G}{$\tilde{T}(t)-\tilde{R}(t)$} 
\centerline{\epsfig{figure=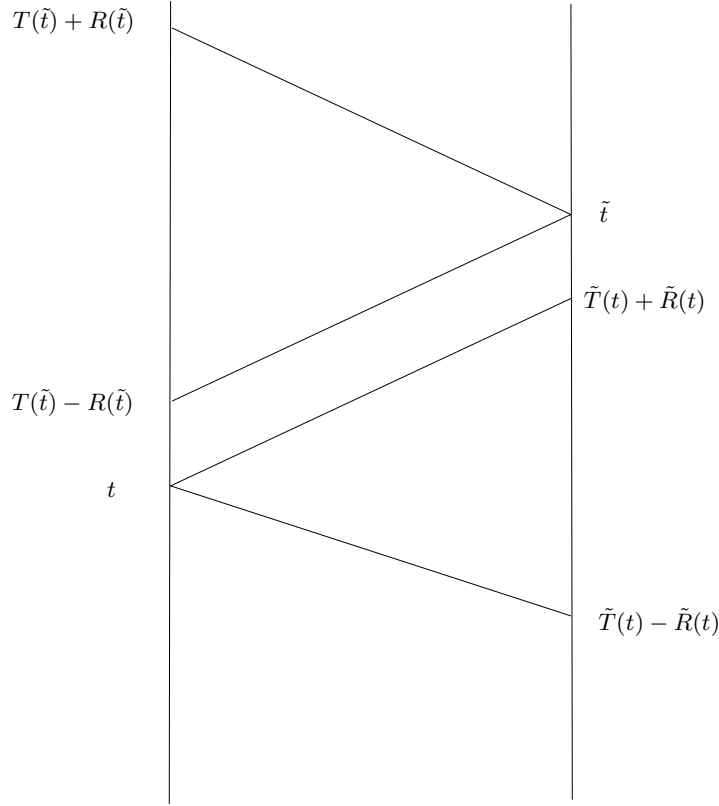, width=8cm}}
\caption{Illustration of the proof of Proposition \protect\ref{prop:symmetry}.}
\label{fig:symmetry}
\end{figure}

We illustrate this result with an example in Minkowski spacetime, using standard
coordinates $(x^0,x^1,x^2,x^3)$ such that the metric takes the form 
\begin{equation}\label{eq:Minkowski}
g=-(dx^0)^2+(dx^1)^2+(dx^2)^2+(dx^3)^2 \: .
\end{equation}
We consider the two clocks 
\begin{equation}\label{eq:clock1}
	\gamma (t)\, = \, ( \, t \, , \, 0 \, , \, 0 \, , \, 0 \, ) \: ,
\end{equation}
\begin{equation}\label{eq:clock2}
	\tilde{\gamma} (\tilde{t}) \, = \, \big( \, \sqrt{1-\omega^2R_0^2}\,  \tilde{t} \, ,
	\, R_0 \, \mathrm{cos} \, \omega \tilde{t} \, , \, R_0 \, 
	\mathrm{sin} \, \omega \tilde{t} \, , \, 0 \, \big) \: ,
\end{equation}
where $\omega$ and $R_0$ are constants such that $\omega^2 R_0^2 <1$.
In both cases the parameter is proper time, i.e., both clocks are standard clocks
with the usual choice of the time unit. The first clock is at rest at the origin of the 
coordinate system, the other clock moves with constant angular velocity $\omega$
on a circle with radius $R_0$ around the origin. An elementary exercise shows
that the radar method carried through with respect to $\gamma$ assigns to each
event $\tilde{\gamma} ( \tilde{t} )$ the time $T(\tilde{t})=  \sqrt{1-\omega^2R_0^2} \, \tilde{t}$
and the distance $R(\tilde{t})=R_0$. On the other hand, the radar method carried
through with respect to $\tilde{\gamma}$ assigns to each event $\gamma (t)$ the
time $\tilde{T} (t) =  t / \sqrt{1-\omega ^2 R_0^2}$ and the distance
$\tilde{R} (t)=  R_0 / \sqrt{1-\omega ^2 R_0^2}$. Thus, neither clock is synchroneous
to the other, and they assign to each other constant but different distances. Now let
us modify this example by changing the time unit for $\tilde{\gamma}$
according to the affine transformation $\tilde{t} \mapsto \hat{t} =
\sqrt{1-\omega ^2 R_0^2} \, \tilde{t}$. This transformation replaces 
$\tilde{\gamma}$ with a new clock $\hat{\gamma}$, 
\begin{equation}\label{eq:hatgamma}
	\hat{\gamma} (\hat{t})= \big( \, \hat{t} \, , \, R_0 \, 
	\mathrm{cos} \, \frac{\omega \, \hat{t}}{\sqrt{1-\omega ^2 R_0^2}} \, ,
	\, R_0 \, \mathrm{sin} \, \frac{\omega \, \hat{t}}{\sqrt{1-\omega ^2 R_0^2}} 
	\, , \, 0 \, \big) \:.
\end{equation}
Note that $\hat{\gamma}$ is still a standard clock, but not with the usual time
unit. We now find that the radar method carried through with respect to $\gamma$ 
assigns to each event $\hat{\gamma} ( \hat{t} )$ the time $T(\hat{t})=  \hat{t}$
and the distance $R(\tilde{t})=R_0$. On the other hand, the radar method carried
through with respect to $\hat{\gamma}$ assigns to each event $\gamma (t)$ the
time $\hat{T} (t) = t$ and the distance $\hat{R} (t)=  R_0$. This modified
example illustrates that Proposition \ref{prop:symmetry} may apply to situations
where there is no symmetry between the two clocks.

\section{Observer fields}
\label{sec:fields}

By an \emph{observer field} on a general-relativistic spacetime we mean a smooth vector field 
$V$ which is everywhere timelike and future-pointing. An observer field $V$ is called
a \emph{standard observer field} if $g(V,V)=-1$. According to our earlier terminology,
integral curves of observer fields are clocks, and integral curves of standard observer
fields are standard clocks with the usual choice of time unit. For the sake of brevity, 
we will refer to the integral curves of an observer field $V$ as to ``clocks in $V$''. 
Note that $V$ fixes the parametrization for each of its integral curves uniquely up to 
an additive constant, i.e., for each clock in $V$ there is still the freedom of ``choosing 
the zero point on the clock's dial''.

In this section we consider the following four properties of an observer field $V$, and for
each of them we give neccesary and sufficient conditions on $V$ under which it is satisfied.

\begin{description}
\item[{\bf Property A}:]
For each clock $\gamma$ in $V$, any other clock in $V$ that is sufficiently close to $\gamma$
such that the radar method can be carried through is synchroneous with $\gamma$, provided 
that the additive constant has been chosen appropriately. 
\item[{\bf Property B}:]
For each clock $\gamma$ in $V$, any other clock in $V$ that is sufficiently close to $\gamma$
such that the radar method can be carried through has temporally constant radar distance 
from $\gamma$. 
\item[{\bf Property C}:]
For any three clocks $\gamma _1$, $\gamma _2$, $\gamma _3$ in $V$ which are sufficiently
close to each other the following is true.
If one light ray from $\gamma _1$ to $\gamma _3$ intersects the worldline of $\gamma _2$,
then all light rays from $\gamma _1$ to $\gamma _3$ intersect the worldline of $\gamma _2$. 
\item[{\bf Property D}:]
For any two clocks $\gamma _1$ and $\gamma _2$ in $V$ that are sufficiently close to each
other, the light rays from $\gamma _1$ to $\gamma _2$ and the light rays from $\gamma _2$ to 
$\gamma _1$ span the same 2-surface. 
\end{description}

All four properties are obviously closely related to the radar method, and we will 
discuss them one by one. In the following we have to assume that the reader is 
familiar with the standard text-book decomposition of the covariant derivative
of an observer field into acceleration, rotation, shear and expansion, and with the
related physical interpretation.

We begin with Property A. We emphasize that in the formulation of this property we restricted
to clocks that are sufficiently close to each other such that the radar method can be carried
through, but \emph{not} to clocks that are infinitesimally close. The synchronizability 
condition for infinitesimally close clocks is a standard text-book matter, see e. g.
Sachs and Wu \cite{SachsWu1977}, Sect. 2.3 and 5.3. One finds that this condition is 
satisfied, for an appropriately rescaled observer field $e^f V$, if and only if $V$ is 
irrotational, i.e. locally hypersurface-orthogonal. The rescaling means that the clocks
of the observers have to be changed appropriately. The synchronization condition for clocks
that are not infinitesimally close to each other is less known. It is given in the following
Proposition.

\begin{proposition}\label{prop:synch}
\begin{itemize}
\item[\emph{(i)}]
A standard observer field $V$ satisfies Property A if and only if $V$ is an irrotational
Killing vector field.
\item[\emph{(ii)}]
An arbitrary $($not necessarily standard$\, )$ observer field $V$ satisfies Property A if and only
if $V$ is an irrotational conformal Killing vector field.
\end{itemize}
\end{proposition}
\begin{proof}
The hard part of the proof is in a paper by Kuang and Liang \cite{KuangLiang1993} who proved
the following. If $V$ is a standard observer field, any point admits a neighborhood that can
be sliced into hypersurfaces that are synchronization hypersurfaces for all clocks in $V$
if and only if $V$ is proportional to an irrotational Killing vector field. In this case, the
flow of the Killing vector field maps synchronization hypersurfaces onto synchronization
hypersurfaces. Clearly, Property A requires in addition that the hypersurfaces can be labeled
such that along each integral curve of $V$ the labeling coincides with proper time. Thus, the
flow of $V$ itself must map synchronization hypersurfaces onto synchronization hypersurfaces.
This completes
the proof of Proposition \ref{prop:synch} (i). Now let $V$ be an arbitrary observer field
on the spacetime $(M,g)$.
Then it is a standard observer field on the conformally rescaled spacetime $(M, -g(V,V)^{-1} g)$.
Clearly, as a conformal factor does not affect the paths of lightlike geodesics, $V$ satisfies
Property A on the original spacetime if and only if it satisfies Property A on the conformally
rescaled spacetime. By Proposition \ref{prop:synch} (i), the latter is true if and only if
$V$ is a normalized irrotational Killing vector field of the metric $-g(V,V)^{-1} g$ and, thus, 
if and only if $V$ is an irrotational conformal Killing vector field of the original metric $g$.
This completes the proof of Proposition \ref{prop:synch} (ii).
\end{proof}

A spacetime that admits an
irrotational Killing vector field normalized to $-1$ is called \emph{ultrastatic}, and a 
spacetime that admits an irrotational conformal Killing vector field is called \emph{conformally
static}. Hence, we can summarize, that ultrastaticity is necessary and sufficient for the 
existence of a standard observer field that satisfies Property A, and conformal staticity is necessary
and sufficient for the existence of a (not necessarily standard) observer field that satisfies
Property A. A simple and instructive example is an expanding Robertson-Walker spacetime. Such 
a spacetime admits a timelike conformal Killing vector field $W$ orthogonal to hypersurfaces
such that $g(W,W)$ is non-constant along the integral curves of $W$. The flow lines of $W$
are often refered to as the ``Hubble flow''. By Proposition \ref{prop:synch} (ii), the observer 
field $W$ satisfies Property A, i.e., if we use on the Hubble flow lines a parametrization adapted
to $W$ (often called ``conformal time''), then the clocks are synchroneous. However, the 
standard observer field $V$ that results by normalizing $W$ does not satisfy Property A, i.e.,
if we use on the Hubble flow lines the parametrization by proper time, the clocks are not 
synchroneous (unless they are infinitesimally close to each other). This example demonstrates 
that it is sometimes mathematically convenient to use non-standard observer fields.

We now turn to Property B which may be viewed as a rigidity condition. Again, 
there is a well-known text-book result on the situation where only clocks that 
are infinitesimally close are considered: For a standard observer field, 
any two clocks that are infinitesimally close to each
other have temporally constant radar distance if and only if $V$ has vanishing
shear and vanishing expansion. This is known as the \emph{Born rigidity} condition, 
refering to a classical paper by Born \cite{Born1909} who introduced this rigidity 
notion in \emph{special relativity}. The differential equations for Born rigid 
observer fields in \emph{general relativity} where first written by Salzmann and 
Taub \cite{SalzmannTaub1954}. They have non-trivial integrability conditions, 
i.e., Born rigid observer fields do not exist on arbitrary spacetimes. The 
following important result is known as the \emph{generalized Herglotz-Noether 
theorem}: If $V$ is a Born-rigid, not hypersurface orthogonal standard observer 
field on a spacetime with constant curvature, then $V$ is proportional 
to a Killing vector field. This was proven by Herglotz \cite{Herglotz1910} and 
Noether \cite{Noether1910} for the case of vanishing curvature (Minkowski spacetime) 
and generalized by Williams \cite{Williams1968} to the case of positive or negative 
curvature (deSitter or anti-deSitter spacetime). As in the case of the synchronization
condition, the rigidity condition for clocks that are not infinitesimally close to
each other is less well known. It is given in the following proposition.

\begin{proposition}\label{prop:rigid}
\begin{itemize}
\item[\emph{(i)}]
A standard observer field $V$ satisfies Property B if and only if $V$ is proportional
to a Killing vector field.
\item[\emph{(ii)}]
An arbitrary $($not necessarily standard$\, )$ observer field $V$ satisfies Property B if and 
only if $V = e^f W$, where $W$ is a conformal Killing vector field and $f$ is a scalar
function that is contant along each integral curve of $V$.
\end{itemize}
\end{proposition}
\begin{proof}
For the proof of Proposition \ref{prop:rigid} (i) we refer to M{\"u}ller zum Hagen 
\cite{Mueller1972}. To prove Proposition \ref{prop:rigid} (ii), let $V$ be an 
arbitrary observer field. As in the proof of
Proposition \ref{prop:synch} (ii), we make use of the fact that $V$ satisfies 
Property B with respect to the metric $g$ if and only if $V$ satisfies 
Property B with respect to the conformally rescaled metric $-g(V,V)^{-1} g$.
The latter is true, by Proposition \ref{prop:rigid} (i), if and only if $V = e^f W$
where $W$ is a Killing vector field of the metric $-g(V,V)^{-1} g$. The latter
condition is true if $W$ is a conformal Killing vector field of the original metric
$g$ and $-g(V,V)^{-1} g(W,W) = e ^{-2f}$ is constant along each
integral curve of $V$. This completes the proof. 
\end{proof}
A spacetime that admits a timelike Killing vector field
is called \emph{stationary} and a spacetime that admits a timelike conformal
vector field is called \emph{conformally stationary}. Hence, we can summarize 
that stationarity is necessary and sufficient for the existence of a 
standard observer field with Property B, and that conformal stationarity is necessary and 
sufficient for the existence of a (not necessarily standard) observer field with 
Property B. As an example, we may again consider the Hubble flow in an expanding
Robertson-Walker spacetime. As with Property A, Property B is satisfied if we use 
conformal time but not if we use proper time.    

We now turn to Property C. This property can be rephrased in the following way.
If, from the position of one clock in $V$, two other clocks in $V$ are seen at
the same spot in the sky (i.e., one behind the other), then this will be true
for all times. In a more geometric wording, Property C requires that the light
rays issuing from any one integral curve of $V$ into the past together with 
the integral curves of $V$ are surface forming. In Hasse and 
Perlick \cite{HassePerlick1988}, observer fields with this property 
were called \emph{parallax free}, and the following 
proposition was proven.

\begin{proposition}\label{prop:parallax}
An observer field $V$ satisfies Property C 
if and only if $V$ is proportional to a conformal Killing vector field.
\end{proposition}

The ``if'' part follows from the well-known fact that the flow of a conformal 
Killing vector field maps light rays onto light rays. The proof of the ``only 
if'' part is more involved, see \cite{HassePerlick1988}. Clearly, Property C 
refers only to the motion of 
the clocks, but not to their ``ticking'' (i.e., not to the parametrization). 
Hence, it is irrelevant whether we consider standard observer fields or 
non-standard observer fields.

Finally we turn to Property D which is a way of saying that 
light rays from $\gamma _1$ to $\gamma _2$
take the same spatial paths as light rays from $\gamma _2$ to $\gamma _1$.
If this property is satisfied, there is a timelike 2-surface between $\gamma _1$
and $\gamma _2$ that is ruled by two families of lightlike geodesics. Note
that \emph{any} timelike 2-surface is ruled by two families of lightlike curves;
in general, however, these will not be geodesics. Foertsch, Hasse and Perlick
\cite{FoertschHassePerlick2003} have shown that a timelike 2-surface is
ruled by two families of lightlike geodesics if and only if its second
fundamental form is a multiple of its first fundamental form. In the
mathematical literature, such surfaces are called \emph{totally umbilic}. 
Some constrution methods and examples of timelike totally umbilic 2-surfaces
are discussed in Foertsch, Hasse and Perlick \cite{FoertschHassePerlick2003}. 
Note that in an arbitrary spacetime totally umbilic 2-surfaces need not
exist. This shows that Property D, which requires such a 2-surface between
any two sufficiently close integral curves of some observer field, is quite
restrictive. A criterion is given in the following proposition.

\begin{proposition}\label{prop:umbilic}
An observer field $V$ satisfies Property D 
if and only if $V$ is proportional to an irrotational conformal Killing vector field.
\end{proposition}
\begin{proof}
The proof of the ``if'' part follows from Foertsch, Hasse and Perlick 
\cite{FoertschHassePerlick2003}, Proposition 3. To prove the ``only if'' part, 
fix any event $p$ and let $\gamma$ be the clock in $V$ that passes through 
$p$. On a neighborhhood of $p$, with the worldline of $\gamma$ omitted, consider
two vector fields $X$ and $Y$ such that the integral curves of $X$ are 
future-pointing light rays and the integral curves of $Y$ are past-pointing
light rays issuing from the worldline of $\gamma$. This condition fixes $X$
and $Y$ uniquely up to nowhere vanishing scalar factors. Property D requires
$X$ and $Y$ to be surface forming and $V$ to be tangent to these surfaces.
The first condition is true, by the well-known Frobenius theorem, if and 
only if the Lie bracket of $X$ and $Y$ is a linear combination of $X$ and 
$Y$, and the second condition is true if and only if $V$ is a linear 
combination of $X$ and $Y$. As a consequence, the Lie bracket of $Y$ and 
$V$ must be a linear combination of $Y$ and $V$, i.e., $Y$ and $V$ must be
surface forming. This proves that $V$ must satisfy Property C. Hence,
by Proposition \ref{prop:parallax}, $V$ must be proportional to a conformal 
Killing vector field. What remains to be shown is that this conformal Killing
vector field is irrotational, i.e., hypersurface orthogonal. To that end we
come back to the observation that $V$ is a linear combination of $X$ and $Y$.
This means that, for any integral curve of $V$ in the considered neighborhood,
light rays from $\gamma$ are seen in the same spatial direction in which
light rays to $\gamma$ are emitted. This is true, in particular, for integral 
curves of $V$ that are infinitesimally close to $\gamma$. Synge \cite{Synge1960}
and, in a simplified way, Pirani \cite{Pirani1965} have shown that this 
``bouncing photon construction'' implies that the connecting vector between
the two worldlines is Fermi transported. This is true for \emph{all} pairs
of infinitesimally neighboring worldlines of $V$ if and only if $V$ is irrotational.
 This completes the proof of Proposition \ref{prop:umbilic}.
\end{proof}

The important result to be kept in mind is that, in a spacetime that is not
conformally stationary, it is impossible to find an observer field that
satisfies any of the four Properties A, B, C, D. This demonstrates that
several features of the radar method which intuitively might be taken for
granted, are actually not satisfied in many cases of interest. 


\end{document}